\documentclass{article}
\usepackage[utf8]{inputenc}
\usepackage{fullpage}
\usepackage{amsthm}
\usepackage{amsmath}
\usepackage{amssymb}
\usepackage{mathtools}
\usepackage[colorlinks,citecolor=blue,bookmarks=true,linktocpage]{hyperref}
\usepackage{aliascnt}
\usepackage[numbered]{bookmark}
\usepackage[capitalise]{cleveref}
\usepackage[backend=biber,url=false,arxiv=abs,maxbibnames=100,isbn=false,sortcites=true]{biblatex}
\usepackage{cleveref}
\usepackage{xcolor}
\usepackage{physics}
\usepackage{tcolorbox}

\usepackage[T1]{fontenc}

\newtheorem{theorem}{Theorem}
\newtheorem{lemma}{Lemma}
\newtheorem{corollary}{Corollary}

\newtheorem{observation}{Observation}

\newcommand{\prb}[1]{\textnormal{\scshape #1}}
\newcommand{\xor}{{\mathbin{\triangle}}}

\newcommand{\floor}[1]{{\lfloor{#1}}\rfloor}

\usepackage{boxedminipage}

\title{Computing diverse pair of solutions for tractable \prb{SAT}\thanks{Supported by JSPS KAKENHI Grant Numbers
JP20H00595, % YK kiban A (arimura) buntan
JP21K11752, % YO kiban C
JP21K17812, % KK wakate
JP22H00513, % YO kiban A (Ono) buntan
JP22K17854, % YI wakate
JP23K28034, % YK kiban B
JP23K24806, % KK kiban B (Horiyama) buntan
JP24H00686, % YK kiban A (ito) buntan
JP24H00697, % YK, YO kiban A (tamaki) buntan
JP24K02901, % YI kiban B (Yusuke Kobayashi) buntan
JP24K21315, % YI chosen (Hirai) buntan
JP24K23847, % TG kensuta
JST ACT-X Grant Number JPMJAX2105, % KK ACT-X
and JST SPRING Grant Number JPMJSP2114. % RS SPRING
}}

\author{
Tatsuya Gima\thanks{Faculty of Information Science and Technology, Hokkaido University. Email: \texttt{\{gima,koba\}@ist.hokudai.ac.jp}} \and
Yuni Iwamasa\thanks{Graduate School of Informatics, Kyoto University. Email: \texttt{iwamasa@i.kyoto-u.ac.jp}} \and
Yasuaki Kobayashi\footnotemark[2] \and
Kazuhiro Kurita\thanks{Graduate School of Informatics, Nagoya University. Email: \texttt{kurita@i.nagoya-u.ac.jp}, \texttt{otachi@nagoya-u.jp}} \and
Yota Otachi\footnotemark[4] \and
Rin Saito\thanks{Graduate School of Information Sciences, Tohoku University. Email: \texttt{rin.saito@dc.tohoku.ac.jp}}
}

\begin{document}

\maketitle

\begin{abstract}
    In many decision-making processes, one may prefer \emph{multiple} solutions to a \emph{single} solution, which allows us to choose an appropriate solution from the set of promising solutions that are found by algorithms.
    Given this, finding a set of \emph{diverse} solutions plays an indispensable role in enhancing human decision-making. 
    In this paper, we investigate the problem of finding diverse solutions of \prb{Satisfiability} from the perspective of parameterized complexity with a particular focus on \emph{tractable} Boolean formulas.
    We present several parameterized tractable and intractable results for finding a diverse pair of satisfying assignments of a Boolean formula.
    In particular, we design an FPT algorithm for finding an ``almost disjoint'' pair of satisfying assignments of a $2$CNF formula.
\end{abstract}

\section{Introduction}
\emph{Diversity of solutions} in optimization problems is an important concept, attracting considerable attention in both practical and theoretical contexts.
In many optimization problems, the primary goal is to find a single (nearly) optimal solution.
However, such an optimal solution may not be advantageous, as optimization problems are nothing more than ``approximation'' of real-world problems.
One possible remedy to this issue is to find multiple \emph{diverse} solutions, where we mean by diverse solutions a set of solutions that are ``different'' from each other.
These diverse solutions provide flexibility for considering various intricate real-world factors that cannot be precisely modeled in optimization problems.

Due to such importance, finding diverse solutions for combinatorial optimization problems is well studied in the literature.
In particular, theoretical aspects of finding diverse solutions have been investigated recently, such as polynomial-time solvability~\cite{HanakaKKO21,HanakaKKLO22,deBergMS23}, approximability~\cite{HanakaK0KKO23,GaoGMKTTY22,DoGN023}, and fixed-parameter tractability~\cite{HanakaKKO21,FominGJPS20,FominGPP021,EibenKM23,BasteFJMOPR22} for many combinatorial optimization problems.

Apart from these results, the pursuit of diverse solutions for the satisfiability problem or, more broadly, for the constraint satisfaction and optimization problems have been explored from both practical and theoretical perspectives.
Crescenzi and Rossi~\cite{CrescenziR02} initiated the study of \prb{Max Hamming Distance SAT}, where the objective is to find two satisfying assignments of the given Boolean formula such that the Hamming distance\footnote{The Hamming distance between two truth assignments is defined as the number of variables that are assigned distinct truth values.} between them is maximized.
They analyzed the (in)approximability of this problem, which will be discussed later in detail.
Angelsmark and Thapper~\cite{AngelsmarkT04} gave exact exponential-time algorithms for \prb{Max Hamming Distance 2SAT} and, more generally, \prb{Max Hamming Distance $(d, \ell)$-CSP} with domain size $d$ and arity $\ell$.
Merkl et al.~\cite{MerklPS23} studied the problem of finding $k$ answers of conjunctive queries (which is equivalent to CSP) from the viewpoint of parameterized complexity.
They showed several parameterized complexity lower and upper bounds for acyclic conjunctive queries in terms of the data, query, and combined complexity.
There are numerous experimental studies on finding $k$ satisfying assignments for SAT/CSP that maximize some metric defined over satisfying assignments~\cite{HebrardHOW05,Nadel11,PetitT15,RuffiniVGKBS19,ZhouLYH23,NikfarjamR0023}.

In this paper, we address the problem of finding diverse solutions for the Boolean satisfiability problem (\prb{SAT} for short).
As \prb{SAT} is already NP-hard (for finding a single satisfying assignment), the problem of finding diverse satisfying assignments for Boolean formulas is NP-hard.
Thus, our target is to find diverse satisfying assignments for \emph{tractable} Boolean formulas.
By the seminal work of Schaefer~\cite{Schaefer78}, it would be reasonable to consider the cases where a given Boolean formula belongs to the classes of $2$CNF formulas, Horn formulas, dual Horn Formulas, or XOR formulas (see \Cref{sec:preli} for details).
We particularly focus on the problem of finding a diverse pair of satisfying assignments, namely \prb{Max Hamming Distance SAT}, as it is already inapproximable in the general case.
As mentioned above, Crescenzi and Rossi~\cite{CrescenziR02} showed the following taxonomy of inapproximability.

\begin{theorem}[\cite{CrescenziR02}]
    Given a Boolean formula $\phi$, \prb{Max Hamming Distance SAT} can be solved in polynomial time if $\phi$ is $01$-valid or even-affine;
    APX-complete if $\phi$ is affine; 
    PolyAPX-complete if $\phi$ is strongly $0$-valid, strongly $1$-valid, Horn, dual Horn, or $2$CNF.
    Otherwise, this problem is NP-hard even for finding a feasible solution.
\end{theorem}

Here, a Boolean formula is said to be
\begin{itemize}\setlength{\itemsep}{0pt} 
    \item \emph{$01$-valid} if it is satisfied by both the ``all-$0$'' assignment (i.e., $\alpha(x)=0$ for all variables $x$) and the ``all-1'' assignment (i.e., $\alpha(x)=1$ for all variables $x$);
    \item \emph{even-affine} if it is an XOR formula such that each XOR clause contains an even number of literals;
    \item \emph{affine} if it is an XOR formula;
    \item \emph{strongly $0$-valid} (resp.\  \emph{strongly $1$-valid}) if it is satisfied by every ``at-most-one-1'' (resp.\  ``at-most-one-0'') assignment, that is, $\alpha(x) = 0$ (resp.\  $\alpha(x) = 1$) for all but at most one variables $x$.
\end{itemize}

The taxonomy of \cite{CrescenziR02} shows a complete picture of (in)approximability of \prb{Max Hamming Distance SAT}.
However, in terms of exact solvability, there is still room for further investigation into their results.
To take a step forward in the setting of exact solvability, we analyze the complexity of \prb{Max Hamming Distance SAT} through the lens of \emph{parameterized complexity}.
In this context, we formalize our problems as follows.
\begin{tcolorbox}
\begin{description}\setlength{\itemsep}{0pt} 
  \item[Problem.] \prb{Diverse Pair of Solutions}
  \item[Input.] A Boolean formula $\varphi$ over variable set $X$.
  \item[Parameter.] $d$
  \item[Task.] Decide if there are two satisfying assignments $\alpha_1, \alpha_2$ of $\varphi$ such that $|\alpha_1 \xor \alpha_2| \ge d$.
\end{description}
\end{tcolorbox}

\begin{tcolorbox}
\begin{description}\setlength{\itemsep}{0pt} 
  \item[Problem.] \prb{Dissimilar Pair of Solutions}
  \item[Input.] A Boolean formula $\varphi$ over variable set $X$.
  \item[Parameter.] $s$
  \item[Task.] Decide if there are two satisfying assignments $\alpha_1, \alpha_2$ of $\varphi$ such that $|\alpha_1 \xor \alpha_2| \ge |X| - s$.
\end{description}
\end{tcolorbox}

Here, $\alpha_1 \xor \alpha_2$ for two truth assignments $\alpha_1$ and $\alpha_2$ is the set of variables in $X$ that are assigned distinct truth values in $\alpha_1$ and $\alpha_2$.
Although these two problems are equivalent to \prb{Max Hamming Distance SAT}, they are significantly different in terms of parameterized complexity.

Our results are summarized in \Cref{tab:results}.
\begin{table*}
    \centering
    \caption{A summary of our results on \prb{Diverse Pair of Solutions} and \prb{Dissimilar Pair of Solutions}.}
    {\renewcommand{\arraystretch}{1.1}
    \begin{tabular}{|c|c|c|}\hline
         & \begin{tabular}{c} \prb{Diverse Pair of Solutions} \\ ($|\alpha_1 \xor \alpha_2| \ge d$) \end{tabular} & 
         \begin{tabular}{c} \prb{Dissimilar Pair of Solutions} \\ ($|\alpha_1 \xor \alpha_2| \ge |X| - s$) \end{tabular}\\\hline
         \raisebox{-1pt}{2SAT}         & \raisebox{-1pt}{W[1]-hard (Thm.~\ref{thm:MaxPoS:2SAT:main}) and XP} & \raisebox{-1pt}{FPT (Thm. \ref{thm:DualPoS:2SAT:main})}  \\\hline
         \raisebox{-1pt}{Horn SAT}    & \raisebox{-1pt}{W[1]-hard (Thm.~\ref{thm:MaxPoS:2SAT:main}) and XP} & \raisebox{-1pt}{NP-hard with $s = 0$ (Thm. \ref{thm:DualPoS:Horn:NPC})}\\\hline
         \raisebox{-1pt}{Dual Horn SAT}& \raisebox{-1pt}{W[1]-hard (Thm.~\ref{thm:MaxPoS:2SAT:main}) and XP} & \raisebox{-1pt}{NP-hard with $s = 0$ (Thm. \ref{thm:DualPoS:Horn:NPC})}  \\\hline
         \raisebox{-1pt}{XOR SAT}      & \raisebox{-1pt}{FPT (Thm.~\ref{thm:MaxPos:XOR:main})} & \raisebox{-1pt}{W[1]-hard (Thm.~\ref{thm:MaxPos:XOR:main}) and XP} \\\hline
    \end{tabular}
    }
    \label{tab:results}
\end{table*}
We observe that \prb{Dissimilar Pair of Solutions} is already hard even on monotone or antimonotone CNF formulas with $s = 0$, where a CNF formula is said to be \emph{monotone} (resp.\  \emph{antimonotone}) if it has no negative literals (resp.\  positive literals).
We also show that \prb{Diverse Pair of Solutions} is W[1]-hard when parameterized by $d$.
This intractability result is established even when the instance is restricted to monotone or antimonotone $2$CNF formulas.
On the positive side, we give a fixed-parameter tractable algorithm for \prb{Dissimilar Pair of Solutions} for $2$CNF formulas.
To this end, we reduce the problem to \prb{Almost 2SAT with Hard Constraints} and then develop an algorithm using a similar idea of the algorithm for \prb{Almost 2SAT} due to \cite{CyganPPW13}.
For XOR formulas, we show that \prb{Max Hamming Distance SAT} is equivalent to the problem of finding a solution of a system of linear equations over $\mathbb F_2$ with maximum Hamming weight, which yields the fixed-parameter tractability and W[1]-hardness of \prb{Diverse Pair of Solutions} and \prb{Dissimilar Pair of Solutions}, respectively.
Finally, we consider the class of intersection of Horn and dual Horn formulas.
In this case, by exploiting the lattice structure induced by the solution space of a given formula, we devise a polynomial-time algorithm for finding $k$ satisfying assignments that maximize the sum of pairwise Hamming distances between them.

\paragraph{Independent work.}
Very recently, Misra, Mittal, and Rai~\cite{misra_et_al:LIPIcs.ISAAC.2024.50} obtained similar results for our problems.
They showed that for affine formulas, \prb{Diverse Pair of Solutions} admits a single-exponential FPT algorithm and a polynomial kernelization and \prb{Dissimilar Pair of Solutions} is W[1]-hard.
For 2CNF formulas, they showed that \prb{Diverse Pair of Solutions} is W[1]-hard.
We would like to mention that the parameterized complexity of \prb{Dissimilar Pair of Solutions} for 2CNF formulas was left open in their work, which is shown to be fixed-parameter tractable in this work.

\section{Preliminaries}\label{sec:preli}

\paragraph{Boolean formulas.}
A \emph{literal} is a variable (\emph{positive literal}) or its negation (\emph{negative literal}).
A \emph{clause} is a disjunction of literals.
A conjunction of clauses is called \emph{conjunctive normal form} (\emph{CNF} for short) \emph{formula}.

A CNF formula is said to be
\begin{itemize}\setlength{\itemsep}{0pt} 
    \item \emph{$k$CNF} if each clause contains at most $k$ literals;
    \item \emph{Horn} if each clause contains at most one positive literal;
    \item \emph{dual Horn} if each clause contains at most one negative literal;
    \item \emph{double Horn} if it is a Horn and dual Horn formula;
    \item \emph{monotone} if it has no negative literals;
    \item \emph{antimonotone} if it has no positive literals.
\end{itemize}
By definition, every antimonotone formula is a Horn formula and every monotone formula is a dual Horn formula.
A Boolean formula is called an \emph{XOR formula} (or \emph{affine}) if it is a conjunction of \emph{XOR clauses}, that is, each clause is of the form $x_{i_1} \oplus x_{i_2} \oplus \cdots \oplus x_{i_k}$, where $x_{i_1}, x_{i_2}, \ldots, x_{i_k}$ are literals and $\oplus$ is the (logical) exclusive-or operator.

According to the taxonomy for Boolean CSP due to Schaefer~\cite{Schaefer78}, the problem of deciding whether a given Boolean formula $\phi$ has a satisfying assignment is polynomial-time solvable when $\phi$ is a $2$CNF, Horn, dual Horn, or XOR formula.

For a truth assignment $\alpha\colon X \to \{0, 1\}$ to a Boolean formula $\varphi$ over variable set $X$, we denote by $\overline{\alpha}$ the truth assignment such that $\overline{\alpha}(x) = 1 - \alpha(x)$ for $x \in X$.
For two truth assignments $\alpha_1$ and $\alpha_2$, we denote by $\alpha_1 \xor \alpha_2$ the set of variables that disagree on these assignments, i.e., $\alpha_1 \xor \alpha_2 = \{x \in X : \alpha_1(x) \neq \alpha_2(x)\}$.

\paragraph{Parameterized complexity.}
In parameterized complexity theory~\cite{CyganFKLMPPS15,DowneyF99,Niedermeier06,FlumG06}, we measure the complexity of computational problems with two parameters.
A \emph{parameterized problem} $L$ consists of pairs of input string $x \in \Sigma^*$ and a parameter $k \in \mathbb N$ (i.e., $L \subseteq \Sigma^* \times \mathbb N$).
A central notion in parameterized complexity theory is \emph{fixed-parameter tractability}.
A parameterized problem $L$ is said to be \emph{fixed-parameter tractable} if, given a string $x$ and parameter $k$, there is an algorithm deciding if $(x, k) \in L$ running in time $f(k)|x|^{O(1)}$, where $f$ is a computable function and $|x|$ is the length of $x$.
Similarly to the classical complexity theory, there is a hierarchy of complexity classes of parameterized problems:
\begin{align*}
    \text{FPT} = \text{W}[0] \subseteq \text{W}[1] \subseteq \dots \subseteq \text{W}[\text{P}] \subseteq \text{XP},
\end{align*}
where the class FPT (resp.\  XP) consists of all parameterized problems that are fixed-parameter tractable (resp.\  admit algorithms with running time $|x|^{f(k)}$ for some computable function $f$).
A well-known conjecture states that these inclusions are strict.
In particular, the problems that are hard in W[1] are considered unlikely to be fixed-parameter tractable.

\section{Hardness for \texorpdfstring{$2$}{2}CNF, Horn, and dual Horn formulas}
We first prove the following theorem.
\begin{theorem}\label{thm:DualPoS:Horn:NPC}
    \prb{Dissimilar Pair of Solutions} is NP-complete even for antimonotone or monotone $3$CNF formulas with $s = 0$.
\end{theorem}

\begin{proof}
The problem obviously belongs to NP\@.
We first prove the statement for antimonotone formulas by presenting a reduction from \prb{Set Splitting}, which is known to be NP-complete~\cite{GareyJ79}.

For a bipartition $(A, B)$ of a finite set $U$, we say that $(A, B)$ \emph{splits} $S \subseteq U$ if $A \cap S \neq \emptyset$ and $B \cap S \neq \emptyset$.
\prb{Set Splitting} is defined as follows: Given a family of sets $\mathcal{F}$ over a universe set $U$, the objective is to decide whether there exists a bipartition of $U$ into $A$ and $B$ such that $(A, B)$ splits $S$ for all $S \in \mathcal{F}$.
This problem remains NP-complete even if all sets in $\mathcal{F}$ have at most three elements~\cite{GareyJ79}.

From an instance $(U, \mathcal{F})$ of \prb{Set Splitting}, we construct a $3$CNF formula $\varphi_{\mathcal{F}}$ with variable set $\{x_v : v \in U\}$ such that for each set $S \in \mathcal F$, $\varphi_{\mathcal{F}}$ contains a clause of the form $\bigvee_{v \in S} \neg x_{v}$.
As $|S| \le 3$ for $S \in \mathcal F$, $\varphi_{\mathcal{F}}$ is $3$CNF and antimonotone.
We claim that $(U, \mathcal{F})$ is a yes-instance if and only if $\varphi_{\mathcal{F}}$ has two satisfying assignments $\alpha_1, \alpha_2$ such that $|\alpha_1 \xor \alpha_2| = |U|$, that is, these two assignments disagree on all variables.

We first prove the forward direction.
Let $(A, B)$ be a partition of $U$ that splits each set in $\mathcal{F}$.
We define a truth assignment $\alpha$ by $\alpha(x_v) = 1$ if $v \in A$; $\alpha(x_v) = 0$ otherwise (i.e., $v \in B$) for $v \in U$.
Since $(A, B)$ splits each set in $\mathcal{F}$, each clause $C$ contains distinct variables $x_u$ and $x_v$ such that $\alpha(x_u) = 1$ and $\alpha(x_v) = 0$.
The antimonotonicity of $\varphi_{\mathcal{F}}$ implies that both $\alpha$ and $\overline{\alpha}$ are satisfying assignments of $\varphi_{\mathcal{F}}$.
As $|\alpha \xor \overline{\alpha}| = |U|$, the forward implication follows.
 
We next prove the converse direction. 
Let $\alpha$ and $\beta$ be two satisfying assignments of $\varphi_{\mathcal F}$ with $|\alpha_1 \xor \alpha_2| = |U|$. 
It is easy to see that $\beta = \overline{\alpha}$.
Let $A = \{v \in U  :  \alpha(x_v) = 1\}$ and $B = \{v \in U  :  \beta(x_v) = \overline{\alpha}(x_v) = 1\}$.
Clearly, $(A, B)$ is a partition of $U$.
Since both $\alpha$ and $\beta$ are satisfying assignments, each clause contains variables $x_u, x_v$ with $\alpha(x_u) = 1$ and $\alpha(x_v) = 0$.
Hence, $(A, B)$ splits each set $S$ in $\mathcal{F}$.

In the case of monotone formulas, we construct $\varphi_{\mathcal F}$ by taking a clause $\bigvee _{v \in S} x_v$ instead of $\bigvee_{v \in S} \neg x_v$ for each $S \in \mathcal F$.
Then, we can prove by the same argument as above.
\end{proof}

The above theorem states that it is hard to find two ``completely different'' satisfying assignments even for Horn or dual Horn formulas.
The following theorem states that it is hard to find two ``slightly different'' satisfying assignments even for $2$CNF, Horn, or dual Horn formulas.

\begin{theorem}\label{thm:MaxPoS:2SAT:main}
    \prb{Diverse Pair of Solutions} is \emph{W[1]}-hard even for monotone or antimonotone $2$CNF formulas.
\end{theorem}

To prove this, we give a parameterized reduction from \prb{Maximum Induced Bipartite Subgraph}.
In this problem, given a graph $G = (V, E)$ and an integer $k$, the goal is to determine whether $G$ has an induced bipartite subgraph of at least $k$ vertices.
This problem is known to be W[1]-hard when parameterized by the solution size $k$~\cite{KhotR02}.

From an instance $(G, k)$ of \prb{Maximum Induced Bipartite Subgraph}, we construct a CNF formula $\varphi_G$ as follows.
We define the variable set $X$ of $\varphi$ as $X = \{x_v : v \in V\}$ and construct the formula $\varphi_G$ by taking the conjunction of $\neg x_u \lor \neg x_v$ for all edges $\{u, v\} \in E$.
Let us note that the formula $\varphi_G$ is antimonotone and has only size-2 clauses.

From a vertex set $U \subseteq V$, we can define a truth assignment $\alpha_U\colon X \to \{0, 1\}$ as its indicator function (i.e., $\alpha_U(x_v) = 1$ if and only if $v \in U$).
The following observation is immediate.

\begin{observation}\label{obs:indicator-is}
    Let $U \subseteq V$.
    Then $U$ is an independent set of $G$ if and only if $\alpha_U$ is a satisfying assignment of $\varphi_G$.
\end{observation}

\begin{lemma}\label{lem:hardness:antimonotone}
    There is an induced bipartite subgraph of $G$ with at least $k$ vertices if and only if $\varphi_G$ has two satisfying assignments $\alpha_1, \alpha_2$ such that $|\alpha_1 \xor \alpha_2| \ge k$.  
\end{lemma}
\begin{proof}
For the forward direction, assume that $G$ has two disjoint vertex sets $A$ and $B$ with $|A| + |B| \ge k$ whose union induces a bipartite subgraph of $G$ with color classes $A$ and $B$.
Consider two truth assignments $\alpha_A$ and $\alpha_B$.
By~\Cref{obs:indicator-is}, they are satisfying assignments of $\varphi_G$ since $A$ and $B$ are independent sets of $G$.
Moreover, 
\begin{align*}
    |\alpha_1 \xor \alpha_2 | = |A \setminus B| + |B \setminus A| = |A| + |B| \geq k,
\end{align*}
as $A \cap B = \emptyset$.
Thus the forward direction follows.

For the other direction, let $\alpha_1, \alpha_2$ be satisfying assignments of $\varphi_G$ with $|\alpha_1 \xor \alpha_2| \ge k$.
By~\Cref{obs:indicator-is}, $A' = \{v \in V : \alpha_1(x_v) = 1\}$ and $B' = \{v \in V : \alpha_2(x_v) = 1\}$ are independent sets of $G$.
Consider two subsets $A \subseteq A'$ and $B \subseteq B'$ of $V$ defined as $A = \{v  :  \alpha_1(x_v)=1, \alpha_2(x_v)=0\}$ and $B = \{v  :  \alpha_1(x_v)=0, \alpha_2(x_v)=1\}$.
The union of these two subsets induces a bipartite subgraph of $G$, as $A$ and $B$ are indeed independent sets of $G$.
We can then observe that 
\begin{align*}
    |A| + |B| = |A \setminus B| + |B \setminus A| = |\alpha_1 \xor \alpha_2 |\geq k,
\end{align*}
as $A \cap B = \emptyset$.
Hence the converse direction follows.
\end{proof}

As a corollary of the construction above, we can show that \prb{Diverse Pair of Solutions} is W[1]-hard for monotone formulas.
Instead of taking the conjunction of $\neg x_u \lor \neg x_v$, we take the conjunction of $x_u \lor x_v$ for all edges $\{u, v\} \in E$.
The formula obtained in this way is denoted by $\overline{\varphi}_G$.

\begin{corollary}\label{cor:hardness:monotone}
    There is an induced bipartite subgraph of $G$ with at least $k$ vertices if and only if $\overline{\varphi}_G$ has two satisfying assignments $\alpha_1, \alpha_2$ such that $|\alpha_1 \xor \alpha_2| \ge k$. 
\end{corollary}
\begin{proof}
Similarly to \Cref{lem:hardness:antimonotone}, we can observe that $U$ is an independent set of $G$ if and only if $\overline{\alpha}_U$ is a satisfying assignment of $\overline{\varphi}_G$.
In the forward direction, we take $\overline{\alpha}_A$ and $ \overline{\alpha}_B$ instead of $\alpha_A$ and $\alpha_B$, respectively.
The other direction is also analogous.
\end{proof}

By \Cref{lem:hardness:antimonotone} and \Cref{cor:hardness:monotone}, we conclude the proof of \Cref{thm:MaxPoS:2SAT:main}.
 
We would like to note that \prb{Diverse Pair of Solutions} belongs to XP when the input Boolean formula is either $2$CNF, Horn, or dual Horn.
To see this, we first choose a set of $d$ variables $X'$ from $X$ and a (partial) assignment $\alpha'$ over $X'$.
The subset $X'$ and the partial assignment are intended to contribute to $|\alpha_1 \xor \alpha_2|$ in such a way that $\alpha_1(x) = \alpha'(x)$ and $\alpha_2(x) = \overline{\alpha}'(x)$ for $x \in X'$.
By fixing the truth values of the variables in $X'$ under $\alpha'$ and $\overline{\alpha}'$, we obtain two Boolean formulas $\varphi'$ and $\varphi''$, respectively.
When $\varphi$ is $2$CNF (resp.\  Horn, and dual Horn), both $\varphi'$ and $\varphi''$ are $2$CNF (resp.\  Horn, and dual Horn) as well.
It is easy to verify that $\varphi$ has two satisfying assignments $\alpha_1$ and $\alpha_2$ with $|\alpha_1 \xor \alpha_2| \ge d$ such that $\alpha_1(x) = \alpha'(x)$ and $\alpha_2(x) = \overline{\alpha}'(x)$ for $x \in X'$ if and only if $\varphi'$ and $\varphi''$ are both satisfiable, which can be determined in polynomial time.
Thus, by trying all of these choices (with $\binom{|X|}{d} \cdot 2^d$ candidates), \prb{Diverse Pair of Solutions} can be solved in time $|X|^{O(d)}$. 
 
\section{Fixed-parameter tractability of \prb{Dissimilar Pair of Solutions} for $2$CNF formulas}
This section is devoted to proving that \prb{Dissimilar Pair of Solutions} is fixed-parameter tractable for $2$CNF formulas.
To this end, we first reduce our problem to \prb{Almost 2SAT} with additional constraints and then give an algorithm to solve the reduced problem using an analogous idea of \cite{CyganPPW13}.

Let $\varphi$ be a $2$CNF formula with variable set $X = \{x_1, \dots, x_n\}$.
We construct a $2$CNF formula $\varphi^*$ as follows.
We first duplicate the same $2$CNF formula with a new variable set $Y = \{y_1, \dots, y_n\}$ and denote it by $\varphi'$.
Then we construct the entire formula $\varphi^*$ as
\begin{align*}
    \varphi^* \coloneqq \varphi \land \varphi' \land \bigwedge_{1 \le i \le n}\left((x_i \lor y_i) \land (\neg x_i \lor \neg y_i)\right).
\end{align*}
We refer to each pair of clauses $(x_i \lor y_i) \land (\neg x_i \lor \neg y_i)$ as \emph{asynchronous clauses}: For any satisfying assignment $\alpha$ of $\varphi^*$, it holds that $\alpha(x_i) \neq \alpha(y_i)$.
Let $S$ be the set of asynchronous clauses of $\varphi^*$.

\begin{lemma}\label{lem:FPT2SAT:async-deletion}
    Let $s$ be a nonnegative integer.
    There are two satisfying assignments $\alpha_1$ and $\alpha_2$ of $\varphi$ with $|\alpha_1 \xor \alpha_2| \ge n - s$ if and only if there are at most $s$ clauses in $S$ whose removal makes $\varphi^*$ satisfiable.
\end{lemma}
\begin{proof}
To prove the forward direction, suppose that there are satisfying assignments $\alpha_1$ and $\alpha_2$ of $\varphi$ such that $|\alpha_1 \xor \alpha_2| \ge n - s$.
We define a $2$CNF formula $\hat{\varphi}$ obtained from $\varphi^*$ by removing each clause of the form $x_i \lor y_i$ when $\alpha_1(x_i)=\alpha_2(x_i) =0$, and each clause of the form $\neg x_i \lor \neg y_i$ when $\alpha_1(x_i)=\alpha_2(x_i) =1$.
Note that for each $1 \le i \le n$, at most one of the pair of asynchronous clauses $(x_i \lor y_i)$ and $(\neg x_i \lor \neg y_i)$ is removed.
As $\hat{\varphi}$ is obtained by removing at most $s$ clauses in $S$ from $\varphi^*$, it suffices to show that $\hat\varphi$ is satisfiable.

We define a truth assignment $\beta$ of $\hat{\varphi}$ as
$\beta(x_i) = \alpha_1(x_i)$ and $\beta(y_i) = \alpha_2(x_i)$ for $1 \le i \le n$.
Since both $\alpha_1$ and $\alpha_2$ are satisfying assignments of $\varphi$, all clauses in $\varphi$ and $\varphi'$ are satisfied by $\beta$.
For each $1 \le i \le n$, the pair of asynchronous clauses $(x_i \lor y_i) \land (\neg x_i \lor \neg y_i)$ is satisfied by $\beta$ if $\alpha_1(x_i) \neq \alpha_2(x_i)$.
Otherwise, we remove clause~$x_i \lor y_i$ (resp.\  clause~$\neg x_i \lor \neg y_i$) when $\beta(x_i) = \beta(y_i) = \alpha_j(x_i) = 0$ (resp.\   $\beta(x_i) = \beta(y_i) = \alpha_j(x_i) = 1$) for $j = 1, 2$.
This implies that the remaining clauses in $\hat\varphi$ are satisfied by $\beta$.
Thus $\hat{\varphi}$ is satisfiable.

To prove the opposite direction, suppose that $\hat{\varphi}$ is a satisfiable formula obtained from $\varphi^*$ by deleting at most $s$ clauses in $S$.
Let $\beta$ be a satisfying assignment of $\hat \varphi$.
We define two truth assignments $\alpha_1, \alpha_2$ by $\alpha_1(x_i) = \beta(x_i)$ and $\alpha_2(x_i) = \beta(y_i)$.
Since $\hat \varphi$ contains $\varphi$ and $\varphi'$ as subformulas, $\alpha_1$ and $\alpha_2$ are satisfying assignments of $\varphi$.
The asynchronous clauses in $\hat\varphi$ ensure that
there are at least $n - s$ pairs of variables $x_i, y_i$ such that $\beta(x_i) \neq \beta(y_i)$, which implies that
\begin{align*}
    |\alpha_1 \xor \alpha_2| &= |\{x_i \in X : \alpha_1(x_i) \neq \alpha_2(x_i)\}|\\
    &= |\{x_i \in X : \beta(x_i) \neq \beta(y_i)\}|\\
    &\ge n - s.
\end{align*}
Thus the lemma follows.
\end{proof}

The lemma enables us to reduce our problem to the following \prb{Almost 2SAT with Hard Constraints}.
In this problem, given a $2$CNF formula $\varphi^*$, a subset $S$ of clauses of $\varphi^*$, and a nonnegative integer $s$, the objective is to determine if there are at most $s$ clauses in $S$ whose removal makes $\varphi^*$ satisfiable.
This problem is a natural extension of the well-known \prb{Almost 2SAT}, which corresponds to the case where $S$ contains all clauses of $\varphi^*$.
Applying a similar reduction due to~\cite{CyganPPW13}, we have the following lemma.

\begin{lemma}\label{lem:FPT2SAT:ALMOST2SATwHCs}
    \prb{Almost 2SAT with Hard Constraints} is fixed-parameter tractable parameterized by $s$.
\end{lemma}
\begin{proof}
We give a reduction to \prb{Vertex Cover Above Maximum Matching}, which is known to be fixed-parameter tractable~\cite{CyganPPW13}.
In this problem, given a graph $G$ with parameter $k$, the goal is to determine whether $G$ has a vertex cover of size at most $k+\mu(G)$, where $\mu(G)$ is the size of a maximum matching of $G$.\footnote{It is well known that the minimum size of a vertex cover of $G$ is at least the maximum size of a matching of $G$ for every graph $G$.}
The basic idea of the reduction is similar to that used in \cite{CyganPPW13}.
The key difference from it is to enlarge gadgets, which force us to delete clauses only from $S$.

Let $I=(\varphi, S, s)$ be an instance of \prb{Almost 2SAT with Hard Constraints}.
By replacing a unit clause $\ell$ to $\ell \vee \ell$, without loss of generality, we can assume that there is no unit clause (i.e., a clause composed of a single literal) in $\varphi$.

We construct a graph $G_I$ from $I$ as follows.
Let $X$ be the set of variables and let $C$ be the set of clauses in $\varphi$.
For $x \in X$, let $n_x$ be the number of occurrences of variable $x$ (which counts both $x$ and $\neg x$), and for $1 \le i \le n_x$, let $c_{x, i}$ be the clause that contains $i$-th occurrence of $x$ or $\neg x$.
For a literal $\ell$, we denote by $v(\ell)$ its variable (i.e., $v(x) = v(\neg x)=x$).
For each literal $\ell \in \{x, \neg x\}$, we define the set of $n_x(s + 1)$ vertices $V(\ell) \coloneqq \{v^j_{\ell, i}  :  1 \leq i \leq n_x, 0 \leq j \leq s\}$.
The vertex set of $G_I$ is defined as the union of all sets $V(x)$ and $V(\neg x)$ for $x \in X$.
The graph $G_I$ contains three types of edges: \emph{variable edges}, \emph{hard-clause edges}, and \emph{soft-clause edges}.
For each $x \in X$, we put a variable edge between every pair of vertices $u\in V(x)$ and $v \in V(\neg x)$, that is, $V(x) \cup V(\neg x)$ induces a complete bipartite graph in $G_I$.
For each clause $c = (\ell \lor \ell') \in C \setminus S$, we put a hard-clause edge between $v^j_{\ell, i}$ and $v^j_{\ell', i'}$ for all $0 \le j \le s$, where $i$ and $i'$ are the indices that satisfy $c_{v(\ell), i} = c_{v(\ell'), i'} = c$.
For each clause $(\ell \lor \ell') \in S$, we put a soft-clause edge between $v^0_{\ell, i}$ and $v^0_{\ell', i'}$, where $i$ and $i'$ are defined as above.
Since $V(x) \cup V(\neg x)$ induces a complete bipartite subgraph, $G_I$ has a perfect matching of size $N \coloneqq \sum_{x \in X}n_x \cdot (s+1)$.
In the following, we show that $I$ is a yes-instance if and only if $G_I$ has a vertex cover of size at most $N + s$.
See \Cref{fig:gadgets} for an illustration.
\begin{figure}
    \centering
    \includegraphics{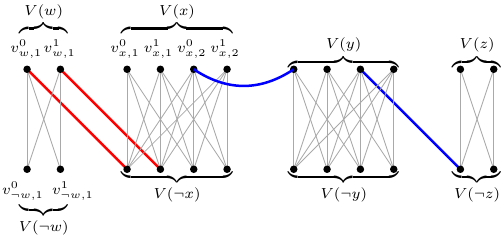}
    \caption{The figure depicts the graph $G_I$ for $\varphi = (w \lor \neg x) \land (x \lor y) \land (y \lor \neg z)$, $S = \{(x \lor y), (y \lor \neg z)\}$, and $s = 1$.
    Red and blue lines represent hard and soft-clause edges, respectively.
    }
    \label{fig:gadgets}
\end{figure}

Suppose that $I=(\varphi, S, s)$ is a yes-instance of \prb{Almost 2SAT with Hard Constraints}, that is, there are at most $s$ clauses in $S$ such that the formula $\hat \varphi$ obtained from $\varphi$ by deleting them is satisfiable.
Let $\alpha$ be a satisfying assignment of $\hat \varphi$.
We define a vertex set $U$ of $G_I$ as follows.
For each variable $x \in X$, the vertex set $U$ contains all vertices in $V(x)$ if $\alpha(x) = 1$ and all vertices in $V(\neg x)$ if $\alpha(x) = 0$.
For each deleted clause $c = (\ell \lor \ell')$, we include vertex $v^0_{\ell, i}$ into $U$, where $c_{v(\ell), i} = c$.
By the construction, we have $|U| \leq N + s$.
We show that $U$ is a vertex cover of $G_I$.
It is easy to see that every variable edge is covered by $U$, as either $V(x) \subseteq U$ or $V(\neg x) \subseteq U$ for $x \in X$.
For each hard-clause edge $e$ corresponding to the clause $(\ell \lor \ell') \in C \setminus S$, $e$ is covered by $U$ as at least one of these literals is evaluated to $1$ under $\alpha$, yielding that the corresponding end vertex is included in $U$.
For each soft-clause edge $e = \{v^0_{\ell, i}, v^0_{\ell', i'}\}$ corresponding to the clause $(\ell \lor \ell') \in S$, $U$ contains at least one of the end vertices of $e$ due to the same argument when it appears in $\hat \varphi$ or due to the fact that $U$ contains $v^0_{\ell. i}$ when it is deleted. 
Consequently, $G_I$ has a vertex cover $U$ with the size at most $N + s$.

To prove the opposite direction, let $U$ be a vertex cover of $G_I$ with $|U| \leq N + s$.
We can observe that at least one of $V(x) \subseteq U$ or $V(\neg x) \subseteq U$ holds for each variable $x$ because otherwise some edge in the complete bipartite graph induced by $V(x) \cup V(\neg x)$ is not covered by $U$.
We define a truth assignment $\alpha$ of $\varphi$ as: For $x \in X$, we set $\alpha(x) = 1$ if $V(\neg x) \cap U = \emptyset$; $\alpha(x) = 0$ if $V(x) \cap U = \emptyset$; and otherwise $\alpha(x) = 1$ or $\alpha(x) = 0$ arbitrarily.
Let $S'$ be the set of all clauses in $\varphi$ that are not satisfied by $\alpha$.
In the following, we prove that (i) $S' \subseteq S$ and (ii) $|S'| \le s$.

Let $U_1$ be the union of $V(\ell)$ for all literals $\ell$ that are evaluated to $1$ under $\alpha$ and let $U_2 = U \setminus U_1$.
Note that $U_1$ and $U_2$ are subsets of $U$.
Moreover, we have $|U_2| = |U| - |U_1| \leq s$ as
\begin{align*}
    |U_1| = \sum_{x \in X}\sum_{\substack{\ell \in \{x, \neg x\}\\\alpha(\ell) = 1}} |V(\ell)| = \sum_{x \in X}n_x \cdot (s+1) = N.
\end{align*}

To show (i) suppose for contradiction that there is a clause $(\ell \lor \ell') \in S' \setminus S$.
There are $s + 1$ hard-clause edges $\{v^j_{\ell, i}, v^j_{\ell', i'}\}$ in $G_I$ with $0 \le j \le s$.
As $(\ell \lor \ell') \in S'$ and it is not satisfied by $\alpha$, we have $v^j_{\ell, i}, v^j_{\ell', i'} \notin U_1$.
Since $U$ is a vertex cover of $G_I$, at least one of $v^j_{\ell, i}$ and $ v^j_{\ell', i'} $ are contained in $U_2$ for each $0 \le j \le s$, which contradicts the fact that $U_2$ contains at most $s$ vertices.

For each clause $(\ell \lor \ell') \in S'$ there is a constraint clause edge $\{v^0_{\ell, i}, v^0_{\ell', i'}\}$ for some $i, i'$.
By the same argument of (i), at least one of its end vertices is contained in $U_2$, implying that $|S'| \le |U_2| \le s$.
\end{proof}

As a consequence of~\Cref{lem:FPT2SAT:async-deletion,lem:FPT2SAT:ALMOST2SATwHCs}, we can reduce our problem to \prb{Vertex Cover Above Maximum Matching}.
By the best known algorithm for \prb{Vertex Cover Above Maximum Matching} due to~\cite{LokshtanovNRRS14}, the following theorem is established.

\begin{theorem}\label{thm:DualPoS:2SAT:main}
    \prb{Dissimilar Pair of Solutions} is solvable in time $2.3146^sn^{O(1)}$, provided that the input formula is restricted to $2$CNF formulas, where $n$ is the number of variables in the input $2$CNF formula.
\end{theorem}

\section{Fixed-parameter tractability and W[1]-hardness for XOR formulas}
In this section, we prove the following upper and lower bound results for XOR formulas.
\begin{theorem}\label{thm:MaxPos:XOR:main}
    \prb{Diverse Pair of Solutions} is fixed-parameter tractable and \prb{Dissimilar Pair of Solutions} is \emph{W[1]}-hard for XOR formulas.
\end{theorem}

It is well known that the satisfiability problem for XOR formulas can be represented as the feasibility problem of linear equations over $\mathbb F_2$ (see~\cite{Schaefer78} for example).
Thus, in the following, we consider the problems of finding solutions $\vb{x}^*_1, \vb{x}^*_2$ of the system of linear equations $A\vb{x} = \vb{b}$ (over $\mathbb F_2$) with $|\vb{x}^*_1 - \vb{x}^*_2| \ge n - s$ (or $|\vb{x}^*_1 - \vb{x}^*_2| \ge d$), where $|\vb{x}|$ is the Hamming weight of a vector $\vb{x}$ and $n$ is the number of variables in the input formula.

By the Gaussian elimination algorithm, we can find a solution $\vb{x}^*$ of $A\vb{x} = \vb{b}$ (if it exists) in polynomial time.
It is well known that each solution $\vb{z}^*$ of $A\vb{x} = \vb{b}$ can be represented as $\vb{z}^* = \vb{x}^* + \vb{y}^*$ for some solution $\vb{y}^*$ of $A\vb{x} = \vb{0}$ and vice versa.
Thus, our problem is equivalent to that of finding two solutions $\vb{y}^*_1, \vb{y}^*_2$ of the system of linear equations $A\vb{y} = \vb{0}$ as
\begin{align*}
    |\vb{x}^*_1 - \vb{x}^*_2| = |\vb{x}^* + \vb{y}^*_1 - (\vb{x}^* + \vb{y}^*_2)| = |\vb{y}^*_1 - \vb{y}^*_2|
\end{align*}
for some solutions $\vb{y}^*_1, \vb{y}^*_2$ of $A\vb{y} = \vb{0}$.
Moreover, as the set of solutions of $A\vb{y} = \vb{0}$ forms a linear space ${\rm Ker} A$, $\vb{y}^*_1 - \vb{y}^*_2$ is also a solution of $A\vb{y} = \vb{0}$ as well.
Given this, it suffices to find a solution $\vb{y}^*$ of $A\vb{y} = \vb{0}$ maximizing its Hamming weight (i.e., $|\vb{y}^*|$).

\begin{observation}\label{obs:XOR:RedtoOneVar}
    Suppose that $A\vb{x} = \vb{b}$ has at least one solution.
    Then there are two solutions $\vb{x}^*_1, \vb{x}^*_2$ of $A\vb{x} = \vb{b}$ with $|\vb{x}^*_1 - \vb{x}^*_2| \ge d$ if and only if there is a solution $\vb{y}^*$ of $A\vb{y} = \vb{0}$ with $|\vb{y}^*| \ge d$.
\end{observation}

The following theorem immediately proves the former part of \Cref{thm:MaxPos:XOR:main}.

\begin{theorem}[\cite{ArvindKKT16}]
    The problem of deciding if there is a solution $\vb{y}^*$ of a given system of linear equations $A\vb{y} = \vb{0}$ with Hamming weight $|\vb{y}^*| \ge d$ is fixed-parameter tractable parameterized by $d$.
\end{theorem}

\Cref{obs:XOR:RedtoOneVar} also proves the latter part of \Cref{thm:MaxPos:XOR:main}.
To see this, we consider the problem of deciding whether a given system of linear equations $A\vb{x} = \vb{0}$ (over $\mathbb F_2$) has a solution of Hamming weight at least $n - s$, where $n$ is the number of columns in $A$.
This problem is known as (the dual parameterized version of) \prb{Even Set} and known to be W[1]-hard parameterized by $s$~\cite{GolovachKS12}.

Without loss of generality, we assume that each row of $A$ contains at least one non-zero component.
For each row $(a_1, \ldots, a_n)$ of $A$, we define an XOR clause $(x_{i_1} \oplus \cdots \oplus x_{i_k})$, where $i_1, \ldots, i_k$ be the indices of the rows with non-zero components.
We then negate an arbitrary one literal, say $x_{i_k}$ for each clause.
From these clauses (with exactly one negative literal each), we define an XOR formula $\varphi$ by taking the conjunction, that is,
\begin{align*}
    \varphi \coloneqq \bigwedge_{\text{row in } A} (x_{i_1} \oplus x_{i_2} \oplus \cdots \oplus \neg x_{i_k}).
\end{align*}
By the above observation, $\varphi$ has a pair of satisfying assignments $\alpha_1, \alpha_2$ with $|\alpha_1 \xor \alpha_2| \ge n - s$ if and only if $A\vb{x} = \vb{0}$ has a solution of Hamming weight at least $n - s$, which proves the latter part of \Cref{thm:MaxPos:XOR:main}.

Note that \prb{Dissimilar Pair of Solutions} for XOR formulas can be solved in time $n^{s + O(1)}$, where $n$ is the number of variables in $\varphi$.
The idea is similar to the one used in~\Cref{lem:FPT2SAT:async-deletion}.
Let $\varphi$ be an input XOR formula with variable set $X = \{x_1, \dots, x_n\}$.
We first guess the candidates of $s$ variables that are allowed to have the same assignment (which can be different) in $\alpha_1$ and $\alpha_2$.
Under this guess, it suffices to find two satisfying assignments $\alpha_1, \alpha_2$ such that $\alpha_1(x) \neq \alpha_2(x)$ for all non-candidate variables $x$. 
To this end we construct a copy $\varphi'$ of $\varphi$ over a new variable set $\{y_1, \dots, y_n\}$ and take $\varphi^* \coloneqq \varphi \land \varphi'$.
For each variable $x_i$ that is not chosen in the first guessing step, we add a clause $(x_i \oplus y_i)$, which enforces that $x_i$ and $y_i$ are assigned different truth values, to $\varphi^*$ by taking a conjunction.
Using an analogous argument in \Cref{lem:FPT2SAT:async-deletion}, the resulting formula $\varphi^*$ is satisfiable if and only if there are two satisfying assignments $\alpha_1$ and $\alpha_2$ of $\varphi$ such that $|\alpha_1 \xor \alpha_2| \ge n - s$ and $\alpha_1$ and $\alpha_2$ are allowed to assign the same truth value only to the variables chosen in the first guessing step.
As $\varphi^*$ is also an XOR-formula, this is decidable in polynomial time.

\section{Polynomial-time algorithm for double Horn formulas}
As seen in the previous sections, \prb{Diverse Pair of Solutions} and \prb{Dissimilar Pair of Solutions} are intractable, and hence we managed to have some positive results through fixed-parameter tractability.
In this section, we restrict our focus on double Horn formulas and give a polynomial-time algorithm for a more general problem: finding $k$ satisfying assignments $\alpha_1, \ldots, \alpha_k$ that maximize $\sum_{1 \le i < j \le k} |\alpha_i \xor \alpha_j|$, which is a common objective in this context~\cite{BasteFJMOPR22,deBergMS23,HanakaKKLO22,HanakaK0KKO23}.

Before proceeding with our algorithm, we first observe that \prb{Diverse Pair of Solutions} is solvable in polynomial time for double Horn formulas.
Let $\varphi$ be a double Horn formula with variable set $X = \{x_1, \ldots, x_n\}$.
Observe that each clause of $\varphi$ forms either a unit clause (i.e., a clause with a single literal) or a clause with exactly one positive literal and exactly one negative literal.
The following well-known algorithm yields a satisfying assignment (if it exists):
\begin{itemize}
    \item[(1)] If $\varphi$ has unit conflict clauses $x_i$ and $\neg x_i$, answer ``NO'' and terminate.
    \item[(2)] If $\varphi$ has a unit clause $x_i$ (resp.\  $\neg x_i$), then we assign $1$ to $x_i$ (resp.\ $0$ to $\neg x_i$) and replace $x_i$ in $\varphi$ with $1$ (resp.\ $0$).
    After this, we remove every clause containing $1$ and then replace every clause $(x \lor 0)$ containing $0$ with a unit clause $x$.
    Repeat (1) and (2) as long as $\varphi$ has a unit clause.
    \item[(3)] Now every clause has exactly one positive literal and exactly one negative literal, and thus we answer ``YES''.
\end{itemize}
In step (3), we can obtain a satisfying assignment by assigning all $1$'s or all $0$'s to the remaining variables.
It is not hard to see that this pair of two assignments $\alpha_1$ and $\alpha_2$ is indeed a solution of \prb{Diverse Pair of Solution} as the unit clause elimination (2) proves that the assignment of $x_i$ is fixed in any satisfying assignments and the assignments $\alpha_1$ and $\alpha_2$ maximize $|\alpha_1 \xor \alpha_2|$.

We extend this through a lattice structure of the set of all satisfying assignments of a double Horn formula $\varphi$.
In the following, we assume that $\varphi$ is satisfiable as otherwise the problem is trivial.
We also use the vector notation $\vb{x}^*$ to represent a particular truth assignment.

Let $\mathcal S \subseteq \{0, 1\}^n$ be the set of satisfying assignments of $\varphi$.
It is known that the solution space of Horn formulas is closed under component-wise AND~\cite{CreignouKS01}, that is, for two satisfying assignments $\vb{x}^* = (x^*_1, \ldots, x^*_n)$ and $\vb{y}^* = (y^*_1, \ldots, y^*_n)$ of a Horn formula $\varphi$,
\begin{align*}
    \vb{x}^* \land \vb{y}^* \coloneqq (x^*_1 \land y^*_1, \ldots, x^*_n \land y^*_n)
\end{align*}
is also a satisfying assignment of $\varphi$.
Symmetrically, the solution space of dual Horn formulas is closed under component-wise OR.
These facts imply that $\mathcal S$ forms a lattice with a natural partial order $\preceq$, that is, $\vb{x}^* \preceq \vb{y}^*$ if and only if $x^*_i \leq y^*_i$ for all $1 \le i \le n$.
Thus the lattice has the unique maximum solution and the unique minimum solution\footnote{These solutions are in fact the two assignments $\alpha_1$ and $\alpha_2$ computed by the algorithm described above.}, which are denoted by $\vb{u}^*$ and $\vb{l}^*$, respectively.
By the distributivity of Boolean algebra, this lattice is, in fact, a distributive lattice,
i.e.,
for all $\vb{x}, \vb{y}, \vb{z} \in \mathcal{S}$,
we have $(\vb{x} \lor \vb{y}) \land \vb{z} = (\vb{x} \land \vb{z}) \lor (\vb{y} \land \vb{z})$.

Now, we make a key observation on the distributive lattice $(\mathcal S, \preceq, \land, \lor)$.
The following lemma is similar to the one used in \cite{deBergMS23}.

\begin{lemma}\label{lem:wHorn:chain}
    Let $k$ be a positive integer.
    Then, there are $k$ satisfying assignments $\vb{x}^*_1, \ldots, \vb{x}^*_k \in \mathcal S$ with $\vb{x}^*_1 \preceq \cdots \preceq \vb{x}^*_k$ that maximize
    \begin{align*}
        \sum_{1 \le i < j \le k} |\vb{x}^*_i - \vb{x}^*_j|
    \end{align*}
    over all combinations of $k$ satisfying assignments.
\end{lemma}
\begin{proof}
Let $f\colon \mathcal S^k \to \mathbb N$ be a function defined as 
\begin{align*}
    f(\vb{x}_1, \ldots, \vb{x}_k) = \sum_{1 \le i < j \le k} |\vb{x}_i - \vb{x}_j|.
\end{align*}
We first see that, for $(\vb{x}_1, \ldots, \vb{x}_k) \in \mathcal{S}^k$ and $i, j \in \{1,\dots, k\}$ with $i < j$, the following identity holds:
\begin{align}
f(\vb{x}_1, \ldots, \vb{x}_k) = f(\vb{x}_1, \ldots,
    \vb{x}_{i-1}, \underline{\vb{x}}_{ij}, \vb{x}_{i+1}, \ldots,
    \vb{x}_{j-1}, \overline{\vb{x}}_{ij},  \vb{x}_{j+1}, \ldots,
    \vb{x}_k),
\label{eq:wHorn:chain:uncrossing}
\end{align}
where $\underline{\vb{x}}_{ij} = \vb{x}_i \land \vb{x}_j$ and $\overline{\vb{x}}_{ij} = \vb{x}_i \lor \vb{x}_j$.
Since the difference of the LHS and RHS of~\eqref{eq:wHorn:chain:uncrossing} is
\begin{align*}
    \sum_{\ell \notin \{i,j\}}\left((|\vb{x}_\ell - \vb{x}_i| +  |\vb{x}_\ell - \vb{x}_j| ) - 
    (|\vb{x}_\ell - \underline{\vb{x}}_{ij}| +  |\vb{x}_\ell - \overline{\vb{x}}_{ij}| )\right)
\end{align*}
and $|\vb{x} - \vb{y}| = |\vb{x} \lor \vb{y}| - |\vb{x} \land \vb{y}|$,
it suffices to see
\begin{align}\label{eq:diff}
    |\vb{x}_\ell \lor \vb{x}_i| - |\vb{x}_\ell \land \vb{x}_i|  
    + |\vb{x}_\ell \lor \vb{x}_j| - |\vb{x}_\ell \land \vb{x}_{j}| = |\vb{x}_\ell \lor \underline{\vb{x}}_{ij}| - |\vb{x}_\ell \land \underline{\vb{x}}_{ij}| 
    + |\vb{x}_\ell \lor \overline{\vb{x}}_{ij}| - |\vb{x}_\ell \land \overline{\vb{x}}_{ij}|
\end{align}
for $\ell \notin \{i,j\}$.
Here we have
\begin{align*}
    |\vb{x}_\ell \lor \vb{x}_i| + |\vb{x}_\ell \lor \vb{x}_j| = |\vb{x}_\ell \lor \overline{\vb{x}}_{ij}| + |(\vb{x}_\ell \lor \vb{x}_i) \land (\vb{x}_\ell \lor \vb{x}_j)| = |\vb{x}_\ell \lor \overline{\vb{x}}_{ij}| + |\vb{x}_\ell \lor \underline{\vb{x}}_{ij}|
\end{align*}
for $\ell \notin \{i,j\}$,
where
the first equality follows from the modularity $|\vb{x}| + |\vb{y}| = |\vb{x} \lor \vb{y}| + |\vb{x} \land \vb{y}|$ of the function $\vb{x} \mapsto |\vb{x}|$
and the second follows from the distributivity of $\mathcal{S}$.
Similarly, we also have
\begin{align*}
    |\vb{x}_\ell \land \vb{x}_i| + |\vb{x}_\ell \land \vb{x}_j| = |\vb{x}_\ell \land \overline{\vb{x}}_{ij}| + |\vb{x}_\ell \land \underline{\vb{x}}_{ij}|.
\end{align*}
Thus we obtain \eqref{eq:diff}.

It is well known (see e.g.,~\cite{Hurkens1988-fs}) in the field of combinatorial optimization that, from any $k$-tuple $(\vb{x}_1, \ldots, \vb{x}_k) \in \mathcal{S}^k$,
we can eventually obtain a totally ordered tuple,
i.e., a tuple $(\vb{y}_1, \dots, \vb{y}_k) \in \mathcal{S}^k$
satisfying $\vb{y}_1 \preceq \cdots \preceq \vb{y}_k$,
by appropriately executing the following procedure finitely many times:
for some $i < j$ with incomparable $\vb{x}_i$ and $\vb{x}_j$, update
\begin{align*}
    (\vb{x}_1, \ldots, \vb{x}_k) \leftarrow (\vb{x}_1, \ldots, \vb{x}_{i-1}, \underline{\vb{x}}_{ij}, \vb{x}_{i+1}, \dots, \vb{x}_{j-1}, \overline{\vb{x}}_{ij}, \vb{x}_{j+1}, \dots, \vb{x}_k).
\end{align*}
By this fact and \eqref{eq:wHorn:chain:uncrossing},
for any $(\vb{x}_1, \ldots, \vb{x}_k) \in \mathcal{S}^k$,
there is a totally ordered tuple $(\vb{y}_1, \dots, \vb{y}_k) \in \mathcal{S}^k$
such that $f(\vb{x}_1, \ldots, \vb{x}_k) = f(\vb{y}_1, \dots, \vb{y}_k)$.
This implies the lemma.
\end{proof}

By~\Cref{lem:wHorn:chain}, we can assume that there is an optimal combination $\vb{x}^*_1, \ldots, \vb{x}^*_k \in \mathcal S$ such that $\vb{x}^*_1 \preceq \cdots \preceq \vb{x}^*_k$.
For $1 \le i \le k$, let $n_i = |\vb{x}^*_i|$.
As $\vb{x}_i \preceq \vb{x}_j$ for $i < j$, we have
\begin{align*}
    |\vb{x_i} - \vb{x_j}| = n_j - n_i.
\end{align*}
Thus we have 
\begin{align*}
     \sum_{i < j} |\vb{x}^*_i - \vb{x}^*_j| &=  \sum_{1 \le i < j \le k} (n_j - n_i)\\
     &=\sum_{1 \le i \le k} \left(\left|\{j : j < i\}\right| - \left|\{j : j > i\}\right| \right) n_i\\
    &= (k - 1)n_k + (k - 3)n_{k-1} + \cdots -(k-3)n_2 - (k-1)n_1\\
    &= \sum_{i = 1}^{\floor{\frac{k}{2}}} (k-(2i-1))(n_{k-i} - n_i).
\end{align*}
Therefore, the objective function attains its maximum when $\vb{x}^*_1 = \vb{x}^*_{2} = \cdots = \vb{x}^*_{\floor{\frac{k}{2}}} = \vb{l}^*$ and $\vb{x}^*_k = \vb{x}^*_{k-1} = \cdots = \vb{x}^*_{k -\floor{\frac{k}{2}}} = \vb{u}^*$, which can be computed in polynomial time.

\begin{theorem}\label{thm:wHorn:polytime}
    The problem of finding $k$ satisfying assignments $\alpha_1, \ldots, \alpha_k$ maximizing $\sum_{1 \le i < j \le k} |\alpha_i \xor \alpha_j|$ is solvable in polynomial time for double Horn formulas.
\end{theorem}

Finally, we would like to remark that the solutions obtained in this section are far from ``diverse solutions'' as it is possible to maximize the objective function with two extreme solutions $\vb{u}^*$ and $\vb{l}^*$.
It might be more interesting to seek solutions that maximize $\min_{1 \le i < j \le k}|\vb{x}_i - \vb{x}_j|$.

\begin{sloppypar}
\printbibliography
\end{sloppypar}

\end{document}